\documentclass[conference,a4paper]{IEEEtran}
\usepackage{amsfonts,amssymb} 
\usepackage{amsmath}
\usepackage{graphicx}
\usepackage{url}

\begin{document}

\newtheorem{theorem}{Theorem}
\newtheorem{lemma}[theorem]{Lemma}
\newtheorem{assumption}{Assumption}


\def\A{ {\mathbb{A} }}
\def\N{ {\mathbb{N} }}
\def\R{ {\mathbb{R} }}

\def\TURN{ { \gamma } }
\def\E{{E}}
\def\P{{P}}
\def\T{{T_{\epsilon,k}}}
\def\T{{T^\epsilon_k}}
\def\U{U^\epsilon}
\def\W{W^\epsilon}
\def\p1{ {g} }
\def\lplus{l^+}
\def\lminus{l^-}
\def\norm{ \eta }

\def\SIZE{ {N_k(l) } }
\def\SIZEp{ {N_k(\lplus) } }
\def\SIZEm{ {N_k(\lminus) } }
\def\SIZEt{ {N_k(l^*) } }
\def\SIZEemax{ {N_k(l^{\epsilon, *, (k)}} }
\def\SIZEmax{ {N_k(l^{*, (k)}) } }

\def\Na{ {n_k(w,a)} }
\def\La{ {L_k} }
\def\Le{ {L_{\epsilon, k}} }
\def\Les{ {L_{\epsilon}} }

\def\lk{l^{(k)}}
\def\lWstar{l^{W}}
\def\lstar{l^*}
\def\lmax{l^{(k)}}
\def\lemax{l^{\epsilon, *, (k)}}

\sloppy

\title{Brute force searching, the typical set and Guesswork}

\author{
   \IEEEauthorblockN{Mark M. Christiansen and Ken R. Duffy}
   \IEEEauthorblockA{Hamilton Institute\\
     National University of Ireland, Maynooth\\
     Email: \{mark.christiansen, ken.duffy\}@nuim.ie} 
  \and
  \IEEEauthorblockN{Fl\'avio du Pin Calmon and Muriel M\'edard }
  \IEEEauthorblockA{Research Laboratory of Electronics\\
Massachusetts Institute of Technology\\
     Email: \{flavio, medard\}@mit.edu}\\
}


\maketitle

\begin{abstract}
Consider the situation where a word is chosen probabilistically
from a finite list. If an attacker knows the list and can inquire
about each word in turn, then selecting the word via the uniform
distribution maximizes the attacker's difficulty, its Guesswork,
in identifying the chosen word. It is tempting to use this property
in cryptanalysis of computationally secure ciphers by assuming coded
words are drawn from a source's typical set and so, for all intents
and purposes, uniformly distributed within it. By applying recent
results on Guesswork, for i.i.d. sources, it is this equipartition ansatz
that we investigate here. In particular, we demonstrate that the
expected Guesswork for a source conditioned to create words in the
typical set grows, with word length, at a lower exponential rate than
that of the uniform approximation, suggesting use of the
approximation is ill-advised.
\end{abstract}

\section{Introduction}

Consider the problem of identifying the value of a discrete random
variable by only asking questions of the sort: is its value X? That
this is a time-consuming task is a cornerstone of computationally
secure ciphers \cite{Menezes96}. It is tempting to appeal to the
Asymptotic Equipartition Property (AEP) \cite{Cover91}, and the
resulting assignment of code words only to elements of the typical
set of the source, to justify restriction to consideration of a
uniform source, e.g. \cite{Pliam00,draper2007,sutcu2008}. This assumed
uniformity has many desirable properties, including maximum
obfustication and difficulty for the inquisitor, e.g. \cite{Calmon12}.
In typical set coding it is necessary to generate codes for words
whose logarithmic probability is within a small distance of the
word length times the specific Shannon entropy.  As a result, while
all these words have near-equal likelihood, the distribution is not
precisely uniform.  It is the consequence of this lack of perfect
uniformity that we investigate here by proving that results on
Guesswork \cite{Arikan96,Malone042,Pfister04,Hanawal11, Christiansen13}
extend to this setting. We establish that for source words originally
constructed from an i.i.d. sequence of letters, as a function of
word length it is exponentially easier to guess a word conditioned
to be in the source's typical set in comparison to the corresponding
equipartition approximation. This raises questions about the wisdom
of appealing to the AEP to justify sole consideration of the uniform
distributions for cryptanalysis and provides alternate results in
their place.

\section{The typical set and Guesswork}

Let $\A=\{0,\ldots,m-1\}$ be a finite alphabet and consider a stochastic
sequence of words, $\{W_k\}$, where $W_k$ is a word of length $k$
taking values in $\A^k$. The process $\{W_k\}$ has specific
Shannon entropy
\begin{align*}
H_W &:= -\lim_{k\to\infty} \frac 1k \sum_{w\in\A^k} \P(W_k=w)\log\P(W_k=w),
\end{align*}
and we shall take all logs to base $e$. For $\epsilon>0$, the typical
set of words of length $k$ is
\begin{align*}
\T := \left\{w\in\A^k: e^{-k(H_W+\epsilon)}\leq P(W_k=w)
			\leq e^{-k(H_W-\epsilon)}\right\}.
\end{align*}
For most reasonable sources \cite{Cover91}, $P(W_k\in\T)>0$ for all
$k$ sufficiently large and typical set encoding results in a new
source of words of length $k$, $\W_k$, with statistics
\begin{align}
\label{eq:defW}
P(\W_k=w) = 
	\begin{cases}
	\displaystyle
	\frac{\P(W_k=w)}{P(W_k\in\T)}
		& \text {if } w\in\T,\\
	0 & \text{if } w\notin \T.
	\end{cases}
\end{align}
Appealing to the AEP, these distributions are often substituted for
their more readily manipulated uniformly random counterpart, $\U_k$,
\begin{align}
\label{eq:defU}
P(\U_k=w) := 
	\begin{cases}
	\displaystyle
	\frac{1}{|\T|}
		& \text {if } w\in\T,\\
	0 & \text{if } w\notin \T,
	\end{cases}
\end{align}
where $|\T|$ is the number of elements in $\T$. While the distribution
of $\W_k$ is near-uniform for large $k$, it is not perfectly uniform
unless the original $W_k$ was uniformly distributed on a subset of
$\A^k$. Is a word selected using the distribution of $\W_k$ 
easier to guess than if it was selected uniformly, $\U_k$?

Given knowledge of $\A^k$, the source statistics of words, say those
of $W_k$, and an oracle against which a word can be tested one at
a time, an attacker's optimal strategy is to generate a partial-order
of the words from most likely to least likely and guess them in
turn \cite{Masey94,Arikan96}. That is, the attacker generates a
function $G:\A^k\to\{1,\ldots,m^k\}$ such that $G(w')<G(w)$ if
$\P(W_k=w')>\P(W_k=w)$. The integer $G(w)$ is the number of guesses
until word $w$ is guessed, its Guesswork.

For fixed $k$ it is shown in \cite{Masey94} that the Shannon
entropy of the underlying distribution bears little relation to the
expected Guesswork, $\E(G(W_k))$, the average number of guesses
required to guess a word chosen with distribution $W_k$ using the
optimal strategy. In a series
of subsequent papers \cite{Arikan96,Malone042,Pfister04,Hanawal11},
under ever less restrictive stochastic assumptions from words made
up of i.i.d. letters to Markovian letters to sofic shifts, an asymptotic
relationship as word length grows between scaled moments of the
Guesswork and specific R\'enyi entropy was identified: 
\begin{align} 
\label{eq:guess} 
\lim_{k\to\infty}\frac 1k \log \E(G(W_k)^\alpha) 
	= \alpha R_W\left(\frac{1}{1+\alpha}\right),
\end{align} 
for $\alpha>-1$, where $R_W(\beta)$ is the specific R\'enyi entropy 
for the process $\{W_k\}$ with parameter $\beta>0$,
\begin{align*} 
R_W(\beta)
	:= \lim_{k\to\infty} \frac 1k 
         \frac{1}{1-\beta}\log\left(\sum_{w\in\A^k} P(W_k=w)^\beta\right).
\end{align*}

These results have recently \cite{Christiansen13} been built on 
to prove that $\{k^{-1}\log G(W_k)\}$ satisfies a Large Deviation
Principle (LDP), e.g \cite{Dembo98}. Define the scaled Cumulant
Generating Function (sCGF) of $\{k^{-1}\log G(W_k)\}$ by
\begin{align*}
\Lambda_W(\alpha):=
	\lim_{k \rightarrow \infty}\frac{1}{k}\log 
	\E\left(e^{\alpha\log G(W_k)}\right)
	\text{ for } \alpha\in\R
\end{align*}
and make the following two assumptions.
\begin{itemize}
\item \emph{Assumption 1:}
For $\alpha>-1$, the sCGF $\Lambda_W(\alpha)$ exists, is equal to
$\alpha R_W\left(1/(1+\alpha)\right)$
and has a continuous derivative in that range.
\item \emph{Assumption 2:}
The limit 
\begin{align}
\label{def:p1}
\p1_W:=\lim_{k\to\infty} \frac 1k \log P(G(W_k)=1) 
\end{align}
exists in $(-\infty,0]$.
\end{itemize}
Should assumptions 1 and 2 hold, Theorem 3 of \cite{Christiansen13} 
establishes that $\Lambda_W(\alpha)=\p1_W$ for all $\alpha\leq-1$
and that the sequence $\{k^{-1}\log G(W_k)\}$
satisfies a LDP with
a rate function given by the Legendre Fenchel transform of the
sCGF, $\Lambda_W^*(x) := \sup_{\alpha\in\R} \{x\alpha
-\Lambda_W(\alpha)\}$. Assumption 1 is motivated by equation
\eqref{eq:guess}, while the Assumption 2 is a regularity condition
on the probability of the most likely word. With
\begin{align}
\label{eq:turn}
\TURN_W
	&:=\lim_{\alpha \downarrow -1}\frac{d}{d\alpha}\Lambda_W(\alpha),
\end{align}
where the order of the size of the set of maximum probability words
of $W_k$ is $\exp(k\TURN_W)$ \cite{Christiansen13}, $\Lambda_W^*(x)$
can be identified as
\begin{align}
	=
	\begin{cases}
	-x-\p1_W & \text{ if } x\in[0,\TURN_W]\\
	\sup_{\alpha\in\R} \{x\alpha -\Lambda_W(\alpha)\} 
		& \text{ if } x\in(\TURN_W,\log(m)],\\
	+\infty & \text{ if } x\notin[0,\log(m)].
	\end{cases}
\label{eq:rf}
\end{align}
Corollary 5 of \cite{Christiansen13} uses this LDP to
prove a result suggested in \cite{Arikan98,Sundaresan07},
that
\begin{align}
\label{eq:shannon}
\lim_{k\to\infty} \frac 1k E(\log(G(W_k))) = H_W,
\end{align}
making clear that the specific Shannon entropy determines the expectation
of the logarithm of the number of guesses to guess the word $W_k$.
The growth rate of the expected Guesswork is a distinct quantity
whose scaling rules can be determined directly from the sCGF in
equation
\eqref{eq:guess},
\begin{align*}
\lim_{k\to\infty} \frac 1k \log E(G(W_k)) = \Lambda_W(1).
\end{align*}
From these expressions and Jensen's inequality, it is clear that
the growth rate of the expected Guesswork is less than
$H_W$. Finally, as a corollary to the LDP, \cite{Christiansen13}
provides the following approximation to the Guesswork distribution
for large $k$:
\begin{align}
\label{eq:wknapprox}
P(G(W_k)=n) \approx 
	\frac 1n \exp\left(-k\Lambda_W^*(k^{-1}\log n)\right)
\end{align}
for $n\in\{1,\ldots,m^k\}$. Thus to approximate the Guesswork
distribution, it is sufficient to know the specific R\'{e}nyi entropy
of the source and the decay-rate of the likelihood of
the sequence of most likely words.

Here we show that if $\{W_k\}$ is constructed from i.i.d.
letters, then both of the processes $\{\U_k\}$ and $\{\W_k\}$ also
satisfy Assumptions 1 and 2 so that, with the appropriate rate
functions, the approximation in equation \eqref{eq:wknapprox} can
be used with $\U_k$ or $\W_k$ in lieu of $W_k$. This enables us to
compare the Guesswork distribution for typical set encoded words
with their assumed uniform counterpart. Even in the simple binary
alphabet case we establish that, apart from edge cases, a word
chosen via $\W_k$ is exponential easier in $k$ to guess on
average than one chosen via $\U_k$.

\section{Statement of Main results}

Assume that the words $\{W_k\}$ are made of i.i.d. letters, defining
$p=(p_0,\ldots,p_{m-1})$ by $p_a=P(W_1=a)$. We shall employ the
following short-hand: $h(l):=-\sum_a l_a\log l_a$ for
$l=(l_0,\ldots,l_{m-1})\in[0,1]^m$, $l_a\geq0$, $\sum_a l_a=1$, so
that $H_W=h(p)$, and $D(l||p):=-\sum_a l_a\log(p_a/l_a)$. Furthermore,
define $\lminus\in[0,1]^m$ and $\lplus\in[0,1]^m$ 
\begin{align}
\label{eq:lminus}
\lminus &\in \arg\max_l\{h(l):h(l)+D(l||p)-\epsilon = h(p)\},\\
\label{eq:lplus}
\lplus &\in \arg\max_l\{h(l):h(l)+D(l||p)+\epsilon = h(p)\},
\end{align}
should they exist. For $\alpha>-1$, also define $\lWstar(\alpha)$ and
$\norm(\alpha)$ by
\begin{align}
\lWstar_a(\alpha) 
	&:=\frac{p_a^{(1/(1+\alpha))}}{\sum_{b \in \A}p_b^{(1/(1+\alpha))}} 
	\text{ for all } a \in \A \text{ and} \label{eq:lWstar} \\
\norm(\alpha)&:= 
	-\sum_a \lWstar_a\log p_a 
	= - \frac{\sum_{a \in \A}p_a^{1/(1+\alpha)}\log p_a}
        {\sum_{b \in \A}p_b^{1/(1+\alpha)}}.
\label{eq:norm}
\end{align}
Assume that $h(p)+\epsilon\leq\log(m)$. If this is not the case,
$\log(m)$ should be substituted in place of $h(\lminus)$ for the
$\{\U_k\}$ results. Proofs of the following are deferred to
the Appendix. 
\begin{lemma}
\label{lem:ass1}
Assumption 1 holds for $\{\U_k\}$ and $\{\W_k\}$ with
\begin{align*}
\Lambda_{\U}(\alpha)&:=
	\alpha h(\lminus)
\end{align*}
and
\begin{align*}
\Lambda_{\W}(\alpha)=\alpha h(\lstar(\alpha))-D(\lstar(\alpha)||p),
\end{align*}
where
\begin{align}
\lstar(\alpha) =
        \begin{cases}
        \lplus & \text { if } \norm(\alpha)\leq-h(p)-\epsilon,\\
        \lWstar(\alpha) & 
		\text { if } \norm(\alpha)\in(-h(p)-\epsilon,h(p)+\epsilon),\\
        \lminus & \text { if } \norm(\alpha)\geq-h(p)+\epsilon.
        \end{cases}
        \label{eq:lstar}
\end{align}
\end{lemma}

\begin{lemma}
\label{lem:ass2}
Assumption 2 holds for $\{\U_k\}$ and $\{\W_k\}$ with
\begin{align*}
\p1_{\U} &= - h(\lminus) \text{ and }\\
\p1_{\W} &= \min\left(-h(p)+\epsilon, \log\max_{a\in\A}p_a\right).
\end{align*}
\end{lemma}
Thus by direct evaluation of the sCGFs at $\alpha=1$,
\begin{align*}
\lim_{k\to\infty} \frac 1k \log E(G(\U_k)) &= h(\lminus) \text{ and }\\
\lim_{k\to\infty} \frac 1k \log E(G(\W_k)) &= \Lambda_{\W}(1).
\end{align*}
As the conditions of Theorem 3 \cite{Christiansen13} are satisfied
\begin{align*}
\lim_{k\to\infty} \frac 1k E(\log(G(\U_k)) &= \Lambda_{\U}'(0) =h(\lminus) \text{ and }\\
\lim_{k\to\infty} \frac 1k E(\log(G(\W_k)) &= \Lambda_{\W}'(0) =h(p),
\end{align*}
and we have the approximations
\begin{align*} 
P(G(\U_k)=n) &\approx \frac 1n \exp\left(-k\Lambda_{\U}^*(k^{-1}\log n)\right)
\text{ and } \\
P(G(\W_k)=n) &\approx \frac 1n \exp\left(-k\Lambda_{\W}^*(k^{-1}\log n)\right).
\end{align*}

\section{Example}

Consider a binary alphabet $\A=\{0,1\}$ and words $\{W_k\}$ constructed
of i.i.d. letters with $P(W_1=0)=p_0>1/2$. In this case there are
unique $\lminus$ and $\lplus$ satisfying equations \eqref{eq:lminus} and
\eqref{eq:lplus} determined by:
\begin{align*}
\lminus_0 &=  p_0-\frac{\epsilon}{\log(p_0)-\log(1-p_0)},\\
\lplus_0 &=  p_0+\frac{\epsilon}{\log(p_0)-\log(1-p_0)}.
\end{align*}
Selecting $0<\epsilon<(\log(p_0)-\log(1-p_0))\min(p_0-1/2,1-p_0)$ ensures
that the typical set is growing more slowly than $2^k$ and that
$1/2<\lminus_0<p_0<\lplus_0<1$.

With $\lWstar(\alpha)$ defined in equation \eqref{eq:lWstar},
from equations \eqref{eq:guess} and \eqref{def:p1}
we have that
\begin{align*}
\Lambda_{W}(\alpha)
	&=
	\begin{cases}
	\log(p_0) & \text{ if } \alpha<-1,\\
	\displaystyle 
	\alpha h(\lWstar(\alpha))-D(\lWstar(\alpha)||p),
		& \text{ if } \alpha\geq-1,\\
	\end{cases}
	\\
	&=
	\begin{cases}
	\log(p_0) & \text{ if } \alpha<-1,\\
	\displaystyle 
	(1+\alpha) \log\left(p_0^\frac{1}{1+\alpha}+(1-p_0)^{\frac{1}{1+\alpha}}\right)
		& \text{ if } \alpha\geq-1,
	\end{cases}
\end{align*}
From Lemmas \ref{lem:ass1} and \ref{lem:ass2} we obtain
\begin{align*}
\Lambda_{\U}(\alpha)
	=
	\begin{cases}
	-h(\lminus) & \text{ if } \alpha<-1,\\
	\alpha h(\lminus) & \text{ if } \alpha\geq-1,
	\end{cases}
\end{align*}
and
\begin{align*}
\Lambda_{\W}(\alpha)=\alpha h(\lstar(\alpha))-D(\lstar(\alpha)||p),
\end{align*}
where $\lstar(\alpha)$ is deinfed in equation \eqref{eq:lstar}
and $\norm(\alpha)$ defined in equation \eqref{eq:norm}.

With $\TURN$ defined in equation \eqref{eq:turn}, we
have $\TURN_W=0$, $\TURN_{\U}=h(\lminus)$ and $\TURN_{\W}=h(\lplus)$
so that, as $h(\lminus)>h(\lplus)$, the ordering of the growth rates
with word length of the set of most likely words from smallest to
largest is: unconditioned source, conditioned source and uniform
approximation.

From these sCGF equations, we can determine the average growth rates
and estimates on the Guesswork distribution. In particular, we
have that
\begin{align*}
\lim_{k\to\infty} \frac 1k E(\log(G(W_k))) &= \Lambda'_{W}(0)= h(p),\\
\lim_{k\to\infty} \frac 1k E(\log(G(\W_k))) &= \Lambda'_{\W}(0)= h(p),\\
\lim_{k\to\infty} \frac 1k E(\log(G(\U_k))) &= \Lambda'_{\U}(0)= h(\lminus).
\end{align*}
As $h((x,1-x))$ is monotonically decreasing for $x>1/2$ and
$1/2<\lminus_0<p_0$, the expectation of the logarithm
of the Guesswork is growing faster for the uniform approximation
than for either the unconditioned or conditioned word source. The
growth rate of the expected Guesswork reveals more features. In
particular, with $A=\norm(1)-(h(p)+\epsilon)$,
\begin{align*}
\lim_{k\to\infty} \frac 1k \log E(G(W_k))  
	&= 2\log(p_0^{\frac{1}{2}} + (1-p_0)^{\frac{1}{2}}),\\
\lim_{k\to\infty} \frac 1k \log E(G(\W_k))  &=
\begin{cases}
2\log(p_0^{\frac{1}{2}} + (1-p_0)^{\frac{1}{2}}),
	A\leq0\\
	h(\lminus)-D(\lminus||p),
		A>0\\
\end{cases}\\
\lim_{k\to\infty} \frac{1}{k} \log E(G(\U_k)) &= h(\lminus).
\end{align*}
For the growth rate of the expected Guesswork, from these it can
be shown that there is no strict order between the unconditioned
and uniform source, but there is a strict ordering between the the
uniform approximation and the true conditioned distribution, with
the former being strictly larger.

With $\epsilon=1/10$ and for a range of $p_0$, these formulae are illustrated
in Figure \ref{fig:fest}. The top line plots 
\begin{align*}
&\lim_{k\to\infty} \frac 1k E(\log(G(\U_k))-\log(G(W_k))) \\
&=\lim_{k\to\infty} \frac 1k E(\log(G(\U_k))-\log(G(\W_k)))
=h(\lminus)-h(p),
\end{align*}
showing that the expected growth rate in the logarithm of the
Guesswork is always higher for the uniform approximation than both
the conditioned and unconditioned sources.  The second highest line
plots the difference in growth rates of the expected Guesswork of
the uniform approximation and the true conditioned source
\begin{align*}
&\lim_{k\to\infty} \frac 1k \log
\frac{E(G(\U_k))}{E(G(\W_k))}\\
&=\begin{cases}
h(\lminus)-2\log(p_0^{\frac{1}{2}} + (1-p_0)^{\frac{1}{2}})
	&\text{ if } \norm(1)\leq h(p)+\epsilon\\
	D(\lminus||p)  
		&\text{ if } \norm(1)>h(p)+\epsilon.
\end{cases}
\end{align*}
That this difference is always positive, which can be established
readily analytically, shows that the expected Guesswork of the true
conditioned source is growing at a slower exponential rate than the
uniform approximation. The second line and the lowest
line,
the growth rates of the uniform and unconditioned expected Guesswork 
\begin{align*}
\lim_{k\to\infty} \frac 1k \log\frac{E(G(\U_k))}{E(G(W_k))}=
h(\lminus)-2\log(p_0^{\frac{1}{2}} + (1-p_0)^{\frac{1}{2}}),
\end{align*}
initially agree. It can, depending on $p_0$ and $\epsilon$, be either
positive or negative. It is negative if the typical set is particularly
small in comparison to the number of unconditioned words.

\begin{figure}
\includegraphics[scale=0.47]{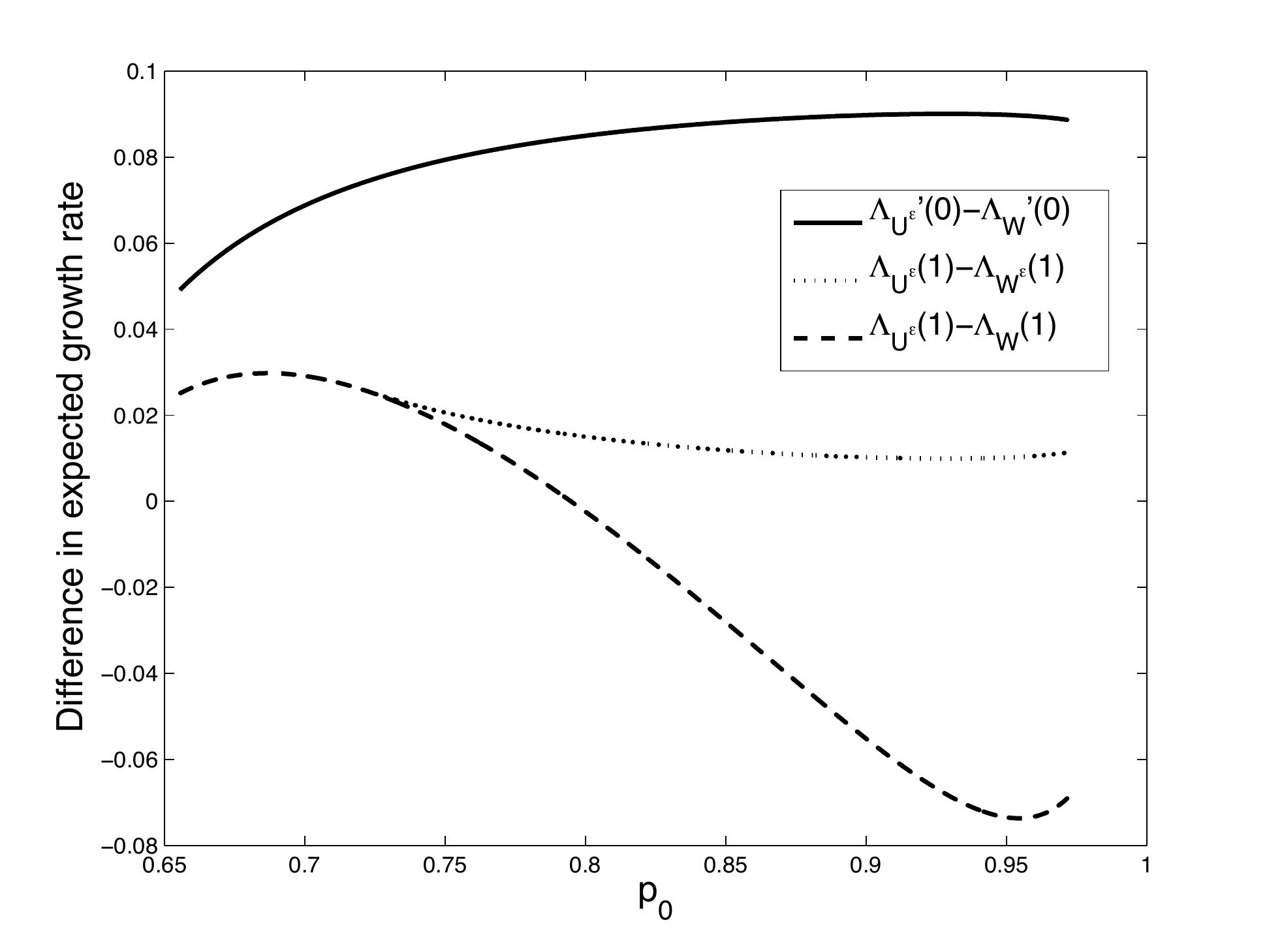}
\caption{Bernoulli$(p_0,1-p_0)$ source. Difference in 
exponential growth rates of Guesswork
between uniform approximation, unconditioned
and conditioned distribution with $\epsilon=0.1$. Top curve is
the difference in expected logarithms
between the uniform approximation and both the conditioned and
unconditioned word sources. Bottom curve is the log-ratio of the
expected Guesswork of the uniform and unconditioned word sources,
with the latter harder to guess for large $p_0$. Middle curve is the log-ratio 
of the uniform and conditioned word sources, which initially follows
the lower line, before separating and staying positive, showing
that the conditioned source is always easier to guess than the
typically used uniform approximation.
}
\label{fig:fest}
\end{figure}

For $p_0=8/10$, the typical set is growing sufficiently slowly that
a word selected from the uniform approximation is easier to guess
than for unconditioned source. For this value, we illustrate the
difference in Guesswork distributions between the unconditioned
$\{W_k\}$, conditioned $\{\W_k\}$ and uniform $\{\U_k\}$ word
sources. If we used the approximation in \eqref{eq:wknapprox}
directly, the graph would not be informative as the range of the
unconditioned source is growing exponentially faster than the other
two. Instead Figure \ref{fig:fest2} plots $-x-\Lambda^*(x)$ for
each of the three processes. That is, using equation \eqref{eq:wknapprox}
and its equivalents for the other two processes, it plots
\begin{align*}
\frac1k \log G(w), \text{ where } G(w)\in\{1,\ldots,2^k\},
\end{align*}
against the large deviation approximations to
\begin{align*}
\frac1k \log P(W_k=w), 
\frac1k \log P(\W_k=w)
\text{ and }
\frac1k \log P(\U_k=w),
\end{align*}
as the resulting plot is unchanging in $k$. The source of the
discrepancy in expected Guesswork is apparent, with the unconditioned
source having substantially more words to cover (due to the log
x-scale). Both it and the true conditioned sources having higher
probability words that skew their Guesswork. The first plateau for
the conditioned and uniform distributions correspond to those words
with approximately maximum highest probability; that is, the length
of this plateau is $\TURN_{\W}$ or $\TURN_{\U}$, defined in equation
\eqref{eq:turn}, so that, for example, approximately $\exp(k\TURN_{\W})$
words have probability of approximately $\exp(k\p1_{\W})$.

\begin{figure}
\includegraphics[scale=0.47]{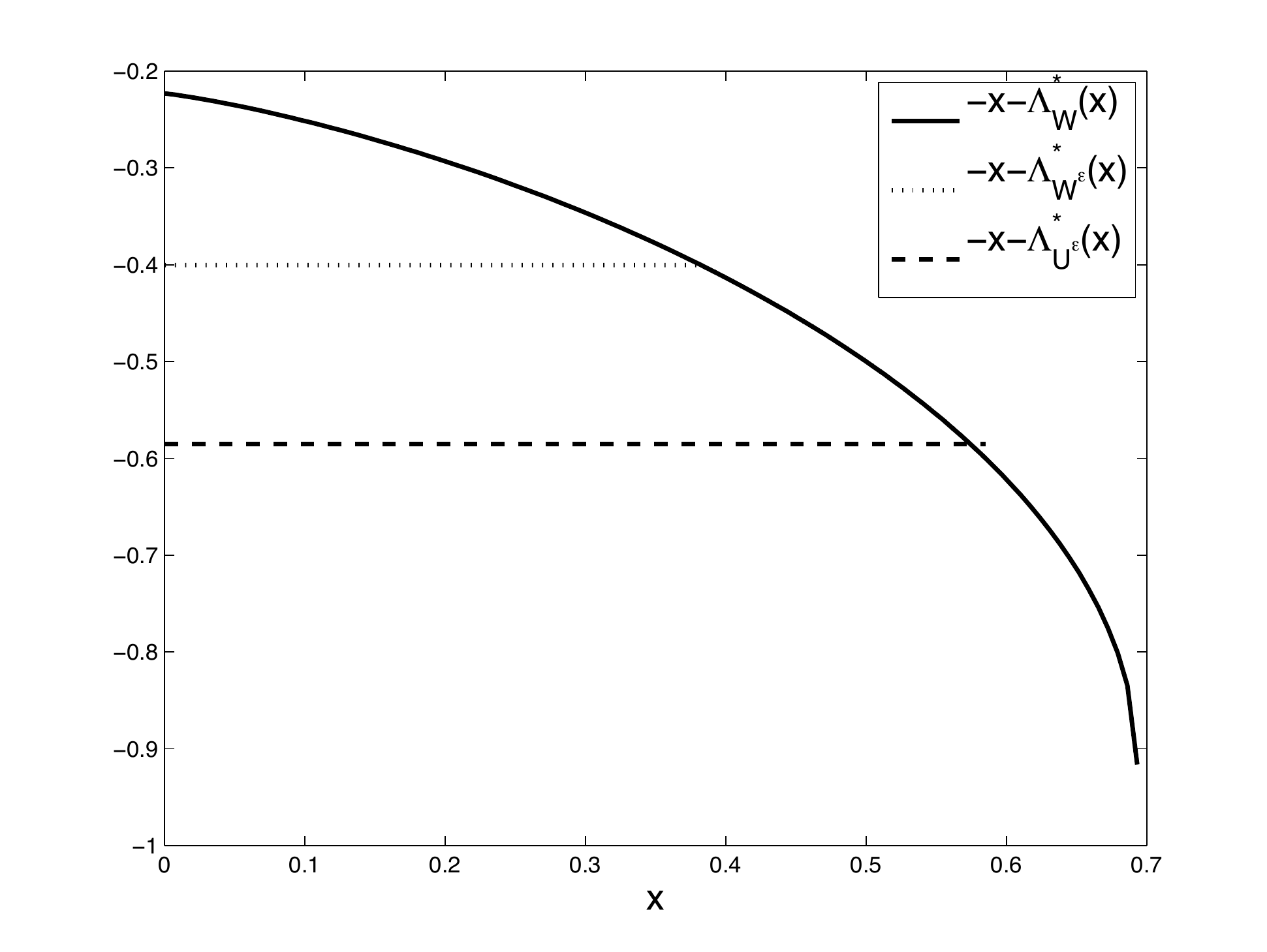}
\caption{Bernoulli($8/10,2/10$) source, $\epsilon=0.1$. 
Guesswork distribution approximations. For large $k$, $x$-axis is
$x=1/k \log G(w)$ for $G(w)\in\{1,\ldots,2^k\}$ and the $y$-axis is
the large deviation approximation $1/k \log P(X=w) \approx
-x-\Lambda_X^*(x)$ for $X=W_k,\W_k$ and $X=\U_k$.
}
\label{fig:fest2}
\end{figure}

\section{Conclusion}

By establishing that the expected Guesswork of a source conditioned
on the typical set is growing with a smaller exponent than its usual
uniform approximation, we have demonstrated that appealing to the
AEP for the latter is erroneous in cryptanalysis and instead provide
a correct methodology for identifying the Guesswork growth rate.

\appendix

The proportion of the letter $a\in\A$ in a word $w=(w_1,\ldots,w_k)\in\A^k$
is given by
\begin{align*}
\Na:=\frac{|\{1 \le i\le k:w_{i}=a\}|}{k}.
\end{align*}
The number of words in a type $l=(l_0,\ldots,l_{m-1})$, where $l_a\geq0$ for all
$a\in\A$ and $\sum_{a\in\A}l_a=1$, is given by  
\begin{align*}
&\SIZE:=|\{w \in \A^k \text{ such that } \Na=l_a \; \forall a \in \A\}|.
\end{align*}
The set of all types, those just in the typical set and smooth
approximations to those in the typical set are denoted
\begin{align*}
\La&:=\{l:\exists w \in \A^k \text{ such that } \Na=l_a \; \forall a \in \A\},\\
\Le&:=\{l:\exists w \in T_{\epsilon, k}
	\text{ such that } \Na=l_a \; \forall a \in \A\},\\
\Les&:=\left\{l:\sum_a l_a \log p_a
	\in[-h(p)-\epsilon,-h(p)+\epsilon]\right\},\\
\end{align*}
where it can readily seen that $\Le\subset\Les$ for all $k$.

For $\{\U_k\}$ we need the following Lemma.
\begin{lemma}\label{lem:T}
The exponential growth rate of the size of the typical set is
\begin{align*}
\lim_{k \rightarrow \infty}\frac 1k \log |\T|=\begin{cases}
&\log m \mbox{ if } \log m\le h(p)+\epsilon\\
&h(\lminus) \mbox{ otherwise}.\end{cases}
\end{align*}
where $\lminus$ is defined in equation \eqref{eq:lminus}.
\end{lemma}
\begin{proof}
For fixed $k$, by the union bound 
\begin{align*}
\max_{l \in \Le}\frac{k!}{\prod_{a \in \A}(kl_a)!}\le |\T|
&\le (k+1)^m\max_{l \in \Le}\frac{k!}{\prod_{a \in \A}(kl_a)!}.
\end{align*}
For the logarithmic limit, these two bounds coincide so
consider the concave optimization problem
\begin{align*}
\max_{l \in \Le}\frac{k!}{\prod_{a \in \A}(kl_a)!}.
\end{align*}
We can upper bound this optimization by replacing $\Le$ with
the smoother version, its superset $\Les$. Using Stirling's
bound we have that
\begin{align*}
&\limsup_{k\to\infty} \frac 1k \log
	\sup_{l \in \Les}\frac{k!}{\prod_{a \in \A}(kl_a)!}\\
&\qquad\leq \sup_{l \in \Les} h(l) 
=\begin{cases}
	\log(m) & \text{ if } h(p)+\epsilon\geq\log(m)\\
	h(\lminus) & \text{ if } h(p)+\epsilon<\log(m).
\end{cases}
\end{align*}

For the lower bound, we need to construct a sequence $\{\lmax\}$
such that $\lmax\in\Le$ for all $k$ sufficiently large and
$h(\lmax)$ converges to either $\log(m)$ or $h(\lminus)$,
as appropriate. Let $l^*=(1/m,\ldots,1/m)$ or $\lminus$
respectively, letting $c\in\arg\max p_a$ and define
\begin{align*}
\lmax_a = 
	\begin{cases}
	k^{-1} \lfloor k\lstar_a \rfloor 
	+ 1-\displaystyle \sum_{b\in\A}\frac1k \lfloor k\lstar_b \rfloor&
	\text{ if } a= c, \\
	k^{-1} \lfloor k\lstar_a \rfloor & \text{ if } a\neq c.
	\end{cases}
\end{align*}
Then $\lmax\in\Le$ for all $k>-m\log(p_c)/(2\epsilon)$ and
$h(\lmax)\to h(\lstar)$, as required.

\end{proof}

\begin{proof}
Proof of Lemma \ref{lem:ass1}. Considering $\{\U_k\}$ first, 
\begin{align*}
\alpha R_{\U}\left(\frac{1}{1+\alpha}\right)
	=\alpha\lim_{k \rightarrow \infty}\frac 1k \log |\T|
	=\alpha h(\lminus),
\end{align*}
by Lemma \ref{lem:T}.
To evaluate $\Lambda_{\U}(\alpha)$,
as for any $n \in \N$ and $\alpha>0$
\begin{align*}
\sum_{i=1}^n i^\alpha \ge \int_0^n x^\alpha dx,
\end{align*}
again using Lemma \ref{lem:T} we have
\begin{align*}
\alpha h(\lminus) 
&=
\lim_{k \rightarrow \infty}\frac 1k\log \frac{1}{1+\alpha}|\T|^\alpha\\
&\le \lim_{k \rightarrow \infty}\frac 1k\log E(e^{\alpha \log G(\U_k)})\\
&=\lim_{k \rightarrow \infty}\frac 1k\log \frac{1}{|\T|} \sum_{i =1}^{|\T|}i^\alpha \\
&\le \lim_{k \rightarrow \infty}\frac 1k \log |\T|^\alpha = \alpha h(\lminus),
\end{align*}
where we have used Lemma \ref{lem:T}. The reverse of these bounds
holds for $\alpha\in (-1,0]$, giving the result.

We break the argument for $\{\W_k\}$ into three steps. Step 1  is
to show the equivalence of the existence of $\Lambda_{\W}(\alpha)$
and $\alpha R_{\W}(1/(1+\alpha))$ for $\alpha>-1$ with the existence
of the following limit
\begin{align}
&\lim_{k \rightarrow \infty}\frac 1k\log \max_{ l \in \Le}\left\{\SIZE^{1+\alpha}\prod_{a \in \A}p_a^{kl_a}\right\}\label{eq:max}.
\end{align}
Step 2 then establishes this limit and identifies it. Step 3 shows
that $\Lambda_{\W}'(\alpha)$ is continuous for $\alpha>-1$. To
achieve steps 1 and 2, we adopt and adapt the method of types
argument employed in the elongated web-version of \cite{Malone042}.

\emph{\underline{Step 1}}  
Two changes from the bounds of \cite{Malone042} Lemma 5.5 are necessary: the
consideration of non-i.i.d. sources by restriction to $\T$; and the
extension of the $\alpha$ range to include $\alpha\in(-1,0]$ from
that for $\alpha\geq0$ given in that document.
Adjusted for conditioning on the typical set we get
\begin{align}
&\frac{1}{1+\alpha}\max_{ l \in \Le}\left\{\SIZE^{1+\alpha}\frac{\prod_{a \in \A}p_a^{kl_a}}{\sum_{w \in \T}P(W_k=w)}\right\}\nonumber\\
&\le E(e^{\alpha \log G(\W_k)}) \le 
\label{eq:upper}
\\
&(k+1)^{m(1+\alpha)}\max_{l\in \Le}\left\{\SIZE^{1+\alpha}\frac{\prod_{a \in \A}p_a^{kl_a}}{\sum_{w \in \T}P(W_k=w)}\right\}.\nonumber
\end{align}
The necessary modification of these inequalities for $\alpha \in (-1, 0]$
gives
\begin{align}
&\max_{ l \in \Le}\left\{\SIZE^{1+\alpha}\frac{\prod_{a \in \A}p_a^{kl_a}}{\sum_{w \in \T}P(W_k=w)}\right\}\nonumber\\
&\le E(e^{\alpha \log G(\W_k)}) \le \label{eq:lower} \\
& \frac{(k+1)^m}{1+\alpha}\max_{l\in \Le}\left\{\SIZE^{1+\alpha}\frac{\prod_{a \in \A}p_a^{kl_a}}{\sum_{w \in \T}P(W_k=w)}\right\}.\nonumber
\end{align}
To show the lower bound holds if $\alpha \in (-1, 0]$ let
\begin{align*}
&\lstar\in \arg\max_{l\in \Le}\left\{\SIZE^{1+\alpha}\frac{\prod_{a \in \A}p_a^{kl_a}}{\sum_{w \in \T}P(W_k=w)}\right\}.
\end{align*}
Taking $\liminf_{k \rightarrow \infty} k^{-1} \log $ and $\limsup_{k
\rightarrow \infty} k^{-1} \log $ of equations \eqref{eq:upper} and
\eqref{eq:lower} establishes that if the limit \eqref{eq:max} exists,
$\Lambda_{\W}(\alpha)$ exists and equals it. Similar inequalities
provide the same result for $\alpha R_{\W}(1/(1+\alpha))$.

\emph{\underline{Step 2}}
The problem has been reduced to establishing the existence of
\begin{align*}
\lim_{k\rightarrow \infty}\frac 1k \log 
\max_{ l \in \Le}\left\{\SIZE^{1+\alpha}\prod_{a \in \A}p_a^{kl_a}\right\}
\end{align*}
and identifying it. The method of proof is similar to that employed
at the start of Lemma \ref{lem:ass1} for $\{\U_k\}$: we provide an upper bound for the limsup
and then establish a corresponding lower bound.

If $\lk\to l$ with $\lk\in\La$, then using Stirling's
bounds we have that
\begin{align*}
\lim_{k\to\infty}\frac1k\log\SIZEmax = h(l). 
\end{align*}
This convergence occurs uniformly in $l$ and so, as $\Le\subset\Les$
for all $k$,
\begin{align}
\limsup_{k\rightarrow \infty}\frac 1k \log 
\max_{ l \in \Le}\left\{\SIZE^{1+\alpha}\prod_{a \in \A}p_a^{kl_a}\right\}
\nonumber\\
\leq\sup_{l\in\Les} \left((1+\alpha)h(l)+\sum_a l_a\log p_a\right)
\nonumber\\
=\sup_{l\in\Les} \left(\alpha h(l)-D(l||p)\right)\label{eq:optim}.
\end{align}
This is a concave optimization problem in $l$ with convex constraints.
Not requiring $l\in\Les$, the unconstrained optimizer over all $l$ is 
attained at $\lWstar(\alpha)$ defined in equation \eqref{eq:lWstar},
which determines $\norm(\alpha)$ in equation \eqref{eq:norm}.
Thus the optimizer of the constrained problem \eqref{eq:optim} can
be identified as that given in equation \eqref{eq:lstar}.
Thus we have that
\begin{align*}
\limsup_{k\rightarrow \infty}\frac 1k \log 
\max_{l \in \Le}\left\{\SIZE^{1+\alpha}\prod_{a \in \A}p_a^{kl_a}\right\}\\
\leq \alpha h(\lstar(\alpha))+D(\lstar(\alpha)||p),
\end{align*}
where $\lstar(\alpha)$ is defined in equation \eqref{eq:lstar}.

We complete the proof by generating a matching lower bound. To do
so, for given $\lstar(\alpha)$ we need only create a sequence such
that $\lk\to\lstar(\alpha)$ and $\lk\in\Le$ for all $k$. If
$\lstar(\alpha)=\lminus$, then the sequence used in the proof of
Lemma \ref{lem:T} suffices. For $\lstar(\alpha)=\lplus$, we use the
same sequence but with floors in lieu of ceilings and the surplus
probability distributed to a least likely letter instead of a most
likely letter. For $\lstar(\alpha)=\lWstar(\alpha)$, either of these
sequences can be used.

\emph{\underline{Step 3}}
As $\Lambda_{\W}(\alpha)=\alpha h(\lstar(\alpha))-D(\lstar(\alpha)||p)$,
with $\lstar(\alpha)$ defined in equation \eqref{eq:lstar},
\begin{align*}
\frac{d}{d\alpha} \Lambda_{\W}(\alpha)
	= h(\lstar(\alpha)) 
	+ \Lambda_{\W}(\alpha) \frac{d}{d\alpha}\lstar(\alpha).
\end{align*}
Thus to establish continuity it suffices to establish continuity
of $\lstar(\alpha)$ and its derivative, which can be done
readily by calculus.
\end{proof}

\begin{proof}

Proof of Lemma \ref{lem:ass2}. This can be established directly by
a letter substitution argument, however, more generically it can be seen
as being a consequence of the existence of specific min-entropy as
a result of Assumption 1 via the following inequalities
\begin{align}
&\alpha R\left(\frac{1}{1+\alpha}\right)-(1+\alpha)\log m\nonumber\\
&\le \liminf_{k \rightarrow \infty} \frac{1+\alpha}{k} \log \frac{m^kP(G(W_k)=1)^{(1/(1+\alpha))}}{m^k}\label{line}\\
&\le\limsup_{k \rightarrow \infty}\frac 1k \log P(G(W_k)=1)\nonumber
\end{align}
\begin{align}
&= (1+\alpha)\limsup_{k \rightarrow \infty}\frac 1k \log P(G(W_k)=1)^{(1/(1+\alpha))}\nonumber\\
&\le (1+\alpha)\limsup_{k \rightarrow \infty}\frac 1k \log (P(G(W_k)=1)^{(1/(1+\alpha))}\nonumber\\
&+\sum_{i=2}^{m^k}P(G(W_k)=i)^{(1/(1+\alpha))})=\alpha R\left(\frac{1}{1+\alpha}\right).\nonumber
\end{align}
Equation \eqref{line} holds as $P(G(W_k)=1) \ge P(G(W_k)=i)$ for all
$i \in \{1, \ldots, m^k\}$. The veracity of the lemma follows as
$\alpha R\left((1+\alpha)^{-1}\right)$ exists and is continuous for
all $\alpha>-1$ by Assumption 1 and $(1+\alpha)\log m$ tends to $0$
as $\alpha \downarrow -1$.

\end{proof}

\section*{Acknowledgment}

M.C. and K.D. supported by the Science Foundation Ireland 
Grant No. 11/PI/1177 and the Irish Higher Educational Authority
(HEA) PRTLI Network Mathematics Grant.  F.d.P.C.  and M.M. 
sponsored by the Department of Defense under Air Force Contract
FA8721-05-C-0002. Opinions, interpretations, recommendations, and
conclusions are those of the authors and are not necessarily endorsed
by the United States Government. Specifically, this work was supported
by Information Systems of ASD(R\&E).



\end{document}